\documentclass[runningheads,11pt]{llncs}
\usepackage{geometry}
\usepackage[T1]{fontenc}
\usepackage{amsmath}
\usepackage{amssymb}
\usepackage{graphicx}
\usepackage{xcolor}
\usepackage{tikz}
\usepackage{subcaption}
 
\usepackage{xspace}
\usepackage{paralist}
\usepackage[normalem]{ulem}
\usepackage{soul}
\setstcolor{green!60!black}
\usepackage[algoruled,lined,algonl,procnumbered]{algorithm2e}
\usepackage[size=scriptsize,color=orange!50!yellow]{todonotes}
\newenvironment{claimproof}
  {\par\noindent\textup{Proof.}\ }
  {\hfill$\square$\par\addvspace{\medskipamount}}
 
\usepackage[hidelinks]{hyperref}
\usepackage{cleveref}
 
\newtheorem{observation}{Observation}
\spnewtheorem{numclaim}{Claim}{\mdseries}{\itshape}
 
\newcounter{claimcase}[numclaim]
\renewcommand{\theclaimcase}{\thenumclaim.\arabic{claimcase}}

\newcounter{subcase}[case]
\renewcommand{\thesubcase}{\thecase.\arabic{subcase}}
\newenvironment{subcase}
  {\par\addvspace{\topsep}\refstepcounter{subcase}\noindent
   \normalfont Case \thesubcase.\ }
  {\par\addvspace{\topsep}}
 
\newenvironment{normalproof}
  {\par\noindent\textbf{Proof.}\ }
  {\hfill$\blacksquare$\par\addvspace{\medskipamount}}

\Crefname{observation}{Observation}{Observations}
\Crefname{numclaim}{Claim}{Claims}
\Crefname{claimcase}{Case}{Cases}
\Crefname{case}{Case}{Cases}
\Crefname{subcase}{Case}{Cases}
 
\newcommand{\exampleG}{\ensuremath{\textsf{G}}\xspace}
\newcommand{\sC}{\ensuremath{\mathsf{C}}\xspace}
\DeclareMathOperator{\parent}{parent}
 
\let\oldnl\nl
\newcommand{\nonl}{\renewcommand{\nl}{\let\nl\oldnl}}

\author{Jérémie Chalopin\inst{1} \and
Maria Kokkou\inst{2}}

\institute{
Aix Marseille Univ, CNRS, LIS, Marseille, France\\
\email{jeremie.chalopin@lis-lab.fr}\and
Paderborn University, Paderborn, Germany\\
\email{maria.kokkou@uni-paderborn.de}
}

\begin{document}

\title{Leader Election via Unique Sink Orientation}
\subtitle{The Case of Anonymous Dismantlable Graphs}
\titlerunning{Leader Election via Unique Sink Orientation}
\maketitle

\begin{abstract}  
  A proof-labeling scheme (PLS) certifies a global property using only local checks. That is, a prover assigns a certificate to every node and each node accepts or rejects by inspecting only the certificates in its neighborhood, so that every graph satisfying some global property admits a certificate assignment accepted by all nodes (existence) and every certificate assignment accepted by all nodes gives a legal configuration (soundness). In this work, we design PLSs for two classical problems, leader certification and spanning tree certification, assuming that all nodes are anonymous, for a specific class of graphs. Certifying a unique leader is impossible in general anonymous graphs  for any certificate size due to a well-known result for rings by [Korman, Kutten and Peleg, 2010]. This raises the question of which graph classes admit a PLS for leader certification with anonymous nodes, ideally with size independent of the size of the graph. Our work contributes to a growing list of such classes, which includes chordal graphs, meshed graphs and simply connected subsets of the triangular grid. We consider dismantlable graphs, a class which contains trees, chordal, bridged and Helly graphs, as well as the $k$-th powers of $\delta$-hyperbolic graphs for $k \geq 2\delta$. To the best of our knowledge, these are the first PLSs tailored to dismantlable graphs, potentially highlighting structural properties useful for certifying additional problems. Finally, we present a generic transformation that converts any labeling scheme with certificates of finite type into a silent self-stabilizing algorithm by adding only one extra state, assuming a Gouda fair scheduler, whereas known generic transformations require $O(\ell + \log n)$ bits for $\ell$-bit certificates. This transformation may be of independent interest. 
\end{abstract}
\keywords{Proof Labeling Scheme, Leader Election, Dismantlable Graphs, Self-Stabilization, Distributed Computing}

\section{Introduction}
A fundamental notion in distributed computing is \emph{locality} \cite{linial1992journal}, which is concerned with two tasks:
\begin{enumerate}
    \item Solving a global problem: If a set of local conditions are satisfied then a global property holds.
    \item Verifying the solution to a global problem: If a global property holds then a set of local conditions are satisfied.
\end{enumerate} 
In this work we focus on verification. That is, we assume that a solution to a given problem is computed globally and the computational entities in the graph need to verify the correctness of this solution locally. Some properties can be verified directly. For instance, in \emph{vertex coloring}, every node knows whether the solution is locally correct by only comparing its color to that of its neighbors at distance one. Such problems are captured by the locally checkable labelings of~\cite{naor1993lcl}. Global properties, however, such as the existence of a unique leader, cannot be verified in this way. \emph{Proof-labeling schemes} (PLSs), introduced by Korman, Kutten and Peleg~\cite{korman2010proof}, address this by allowing a \emph{prover} to assign additional information, called a \emph{certificate}, to every node. Each node then accepts or rejects by checking only the certificates in its neighborhood so that every graph of the considered class admits a certificate assignment accepted by all nodes and every certificate assignment accepted by all nodes gives a legal configuration. 

It is known that on trees, a unique leader can be verified using constant-size certificates. Furthermore, in rings, verifying the existence of a unique leader requires $\Omega(\log n)$-bit certificates, where $n$ is the size of the ring, under the assumption that nodes have unique identifiers. For the case where nodes do not have unique identifiers, a unique leader cannot be certified for any size of the certificates in arbitrary graphs due to the impossibility result of \cite{korman2010proof} for rings. The main motivation of this work was determining a large class of graphs where it is possible to design a PLS for leader certification that does not depend on the size of the graph but instead on some other parameter. We surprisingly prove that there exist  classes of graphs for which it is possible: chordal and dismantlable graphs. For each of these classes we define a set of local constraints and we show that any labeling satisfying the constraints guarantees a unique leader.

Self-stabilization \cite{dijkstra1974stabilising}, a generic way to model various faults and guarantee the correct behavior of an algorithm regardless of transient faults, is naturally connected to the concept of locality as observed for the first time in \cite{afek1997local}. An algorithm is self-stabilizing if it eventually converges to a correct output from any arbitrary initial configuration (see also \cite{dolev2000book}). A self-stabilizing algorithm is called \emph{silent} \cite{dolev1999silent} if when it reaches a correct configuration, it remains in that configuration unless a fault occurs. In a silent self-stabilizing algorithm, each node usually maintains (i) its part of the global solution and (ii) additional information that allows the verification of the global configuration via local checks, which is precisely a certificate in the sense of PLSs. For instance, in the self-stabilizing spanning tree algorithm of~\cite{chen1991spanning}, each node stores a pointer to its parent (as a way to solve the task) and its distance from the root (as a certificate). Without the latter, cycles cannot be locally detected. 

The certificates we define in this work are an orientation of the edges incident to each node. For each class of graphs we consider, we define a set of local constraints and we show that if all constraints we define for each family of graphs are satisfied at every node, the global configuration contains a unique sink that we define to be a \emph{leader}. Then we use this description of what a solution looks like locally to construct a silent self-stabilizing algorithm for leader election in specific classes of graphs that works under a Gouda fair scheduler \cite{Gouda2001theory} which is the strongest kind of fairness \cite{dubois2011taxonomy}. Our self-stabilizing algorithm is generic and works for any problem for which the solution can be defined as a labeling scheme. Results in a similar flavor that combine local checks \cite{afek1997local} with a global ``reset'' \cite{awerbuch1994reset} to some predetermined state in order to achieve self stabilization already exist \cite{awerbuch1994stabilization}. However, contrary to algorithms such as \cite{awerbuch1994reset,awerbuch1994stabilization,BlinFP14}, thanks to Gouda fairness, our algorithm only needs one additional bit of memory to transform a labeling scheme using a finite number of labels
to a self-stabilizing algorithm.

We focus on dismantlable graphs. In a graph $G$, a vertex $v$ is
\emph{dominated} by $u$, if the closed neighborhood $N[v]$ of $v$ is
included in the closed neighborhood $N[u]$ of $u$. If we can reach the
1-vertex graph by successively removing dominated vertices, $G$ is
dismantlable.  Dismantlable graphs are exactly the cop-win graphs
introduced by Quilliot~\cite{quilliot1978jeux} and Nowakowski and
Winkler~\cite{nowakowski1983vertex} in their studies of the classical
Cops and Robbers game~\cite{aigner1984cops}. As a warm-up, we consider
trees and chordal graphs, moved to the appendix due to space
constraints. Trees, dismantlable graphs and chordal graphs can be
characterized by an ordering of their vertices. Trees admit a
\emph{leaf elimination ordering}, in which every vertex has at most one
neighbor appearing earlier in the ordering.  Chordal graphs admit a
\emph{simplicial ordering}, in which for each vertex $v$, the
neighbors of $v$ appearing earlier in the ordering form a complete
graph. Similarly, dismantlable graphs admit a \emph{dismantling order}, where for each vertex $v$, the graph induced by its
neighbors appearing earlier in the ordering contains a universal
vertex $u$. Observe that a leaf elimination ordering is a particular simplicial ordering, and a simplicial ordering is a particular dismantling order.

\subsection{Related Work}\label{sec:rel-work}
The concept of locally verifiable problems was formalized by Naor and Stockmeyer~\cite{naor1995journal}, who defined the family of Locally Checkable Labeling (LCL) problems. A generalization of LCL, called Locally Checkable Proofs (LCP), was introduced in~\cite{goos2016locally}. The main difference between the two approaches is that the goal in LCL is to locally solve a global problem, whereas in LCP (and in the related model of Proof Labeling Schemes (PLS)~\cite{korman2010proof}) the solution can be computed globally and is only verified locally, using a certificate given to each node. In the classical PLS model, each node only has access to the certificates of its neighbors at distance one and the nodes do not have access to the identifiers and input of their neighbors. Some recent work (e.g., \cite{feuilloley2026proving}) uses the term $r$-local PLS to refer to nodes having access to the certificates of their neighbors at distance $r$. When considering only anonymous nodes and local conditions that do not depend on the inputs of the neighbors, as in this paper, the concepts of LCP and $r$-local PLS can be used interchangeably. A thorough introduction to local certification can be found in \cite{feuilloley2019introduction}.

The most similar paper to our work is \cite{chalopin2024stabilising},
which also checks a set of local constraints that ensure the existence
of a unique sink in a specific class of graphs. The main similarity
between our work and \cite{chalopin2024stabilising} is that nodes are
anonymous. In \cite{chalopin2024stabilising}, the authors propose a
constant-memory labeling and a set of sufficient conditions to ensure
a unique sink for a simply connected subset of the regular triangular
grid $S$, where \emph{simply connected} means that the configuration
does not contain ``holes'' (i.e., subsets of the triangular grid that
do not belong in $S$ and are disconnected to each other). As simply
connected parts of triangular grids are dismantlable graphs, this work
(in particular \Cref{sec:k4-free}) generalizes the result
of~\cite{chalopin2024stabilising} and extends this idea to additional
graph classes. The natural generalization of
\cite{chalopin2024stabilising} would be considering graphs with a
contractible clique-complex, but this class is difficult to handle
algorithmically, as recognizing these graphs is undecidable (this
follows from a result of Novikov, as explained
in~\cite{Tancer2016collapsible}). Dismantlable graphs form a class of
graphs that have contractible clique-complexes, but they are easier to
handle as they can be recognized in polynomial
time~\cite{nowakowski1983vertex,quilliot1978jeux}, making them a
natural class to consider.  Observe that adding a universal vertex
(i.e., a vertex neighboring to every other vertex in the graph) to any
$n$-node graph results in an $(n+1)$-node dismantlable graph implying
that the number of dismantlable graphs with $n+1$ vertices is larger
than the number of $n$-node graphs.  Chordal graphs form a subclass of
dismantlable graphs that is already very large as it contains the
class of split graphs that already contains $2^{\Theta(n^2)}$
graphs~\cite{BRW-split1985,oeis-split}.  The class of dismantlable
graphs also includes the class of bridged
graphs~\cite{AnFa,Ch_bridged} that generalize chordal graphs. Bridged
graphs are the graphs where all isometric cycles are of length
3. Examples of bridged graphs are trees, chordal graphs, and
triangular grids (the graphs considered
in~\cite{chalopin2024stabilising} are particular planar bridged
graphs).  Bridged graphs admit several different characterizations
based on their metric and convexity properties~\cite{FaJa,SoCh}. In
particular, bridged graphs are exactly the dismantlable graphs without
induced cycles of length 4 and 5~\cite{AnFa}.
Helly graphs are also dismantlable graphs~\cite{BP-absolute,BaPr}. A
graph is Helly if any family of pairwise intersecting balls has a
non-empty intersection. Examples of Helly graphs are trees and king
grids. Helly graphs are universal in the sense that any graph $G$ is
an isometric subgraph of a Helly graph~\cite{Isbell64,Pesch87}.
Another subclass of dismantlable graphs is the class of visibility
graphs of simple polygons~\cite{LubiwSV-CG17}. Visibility graphs are
important objects in computational geometry and robot motion
planning. They have many interesting combinatorial
properties~\cite{Ghosh} even if no combinatorial characterization is
known for them so far.
Another way to get dismantlable graphs is to consider the
powers\footnote{The $k$th power $G^k$ of a graph $G$ is the graph
  obtained from $G$ by adding an edge between any two vertices at
  distance at most $k$ in $G$.} of $\delta$-hyperbolic graphs.
Initially introduced by Gromov in the context of geometric group
theory~\cite{Gromov}, the hyperbolicity $\delta$ of a graph $G$
measures how close the graph $G$ is from a tree, from a metric point
of view. It has received considerable attention in graph theory as
many real-world graphs have small
hyperbolicity~\cite{KeSN16,NaSa,ShavittT08}. Chepoi and Estellon
proved that for any $\delta$-hyperbolic graph $G$, for any
$k \geq 2\delta$, the $k$th power $G^k$ of $G$ is
dismantlable~\cite{ChEs}.

\begin{figure}[h]
    \centering
    \includegraphics[scale=.55]{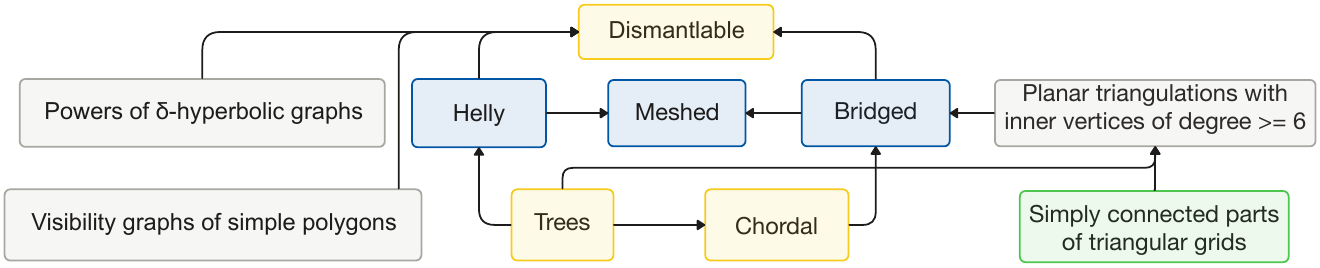}
    \caption{A graph class hierarchy, where each arrow represents an
      inclusion relation. The green class is considered
      in~\cite{chalopin2024stabilising}. The blue classes are the
      classes explicitly considered in~\cite{ChChKo-meshed}. The
      yellow classes are the classes explicitly considered in this
      paper.}
    \label{fig:discussion}
  \end{figure}
  
  Recently, \cite{feuilloley2026proving} studied the verification
  aspect of a unique leader existing in anonymous grids and chordal
  graphs, giving an $O(\log D)$ distance based certification scheme,
  where $D$ is the diameter of the network. After the first version of
  the current paper appeared, the idea of \cite{feuilloley2026proving}
  to verify the existence of a unique leader by a certificate encoding
  the distance of each node to the leader was extended to more classes
  of graphs in \cite{ChChKo-meshed}. Notably, \cite{ChChKo-meshed}
  verifies the existence of a unique leader in meshed graphs using
  certificates of constant size, where the certificate of each node is
  the modulo 3 distance to the leader. Meshed graphs is a class of
  graphs that also contains chordal, bridged, and Helly
  graphs. However, the classes of meshed and dismantlable graphs are
  not comparable.  The relations between the classes considered in
  this paper and in~\cite{ChChKo-meshed,chalopin2024stabilising} are
  summarized in Figure~\ref{fig:discussion}.  In fact, meshed graphs
  are defined by a distance condition whereas dismantlable graphs do
  not admit such a metric characterization. As a result it is not
  clear whether a similar distance based labeling can be defined for
  dismantlable graphs. In this work, the labels we assign to each node
  encode an orientation that implies the existence of a unique sink,
  defined to be a leader. Hence, the size of our labels depends on the
  maximum degree of the graph. This work as well as the work of
  \cite{feuilloley2026proving} and \cite{ChChKo-meshed} assume that
  nodes have access to information that fits within the radius two
  verification rules.

  \subsection{Our Contributions}
  In this work, we design proof-labeling schemes in which the
    certificates are usually edge orientations and in which every
    node checks a set of local constraints on these orientations,
    with the constraints depending on the graph class. 
    Our constraints refer to the graph induced by the closed
    neighborhood of each node (which fits within the usual radius 2
    verification rules) in dismantlable graphs and related classes.
    We show that when the certificates satisfy our constraints at
    every node, the network contains a unique sink that we define to
    be a leader. The size of the certificate at each node is
    \emph{deg} bits (where \emph{deg} is the degree of the node). As
    most of our certificates are an orientation of all edges incident
    to each node, our approach results in constant size certificates
    for graphs of bounded degree. 
    In this work we only consider
    anonymous nodes. Instead of the usual approach of using structural
    properties for the construction of the certificates, in this work we
    capture the geometry of graphs in the certificates. This geometric
    approach can be useful in practical geometric settings, as for
    example our result holds for planar triangulations\footnote{Here,
      a planar triangulation is a planar graph whose all inner faces
      are triangles.} where inner nodes have degree at least six. 

    As mentioned above, dismantlable graphs (respectively,
      chordal graphs) admit a dismantling order (respectively, a
      simplicial order) of their vertices and we exploit this order
      to show the existence of an orientation satisfying our
      constraints.  However, the orientations obtained in this way are
      acyclic, while the orientations satisfying our constraints form
      a much larger class of orientations, as they may contain large
      cycles. We thus need to show that \emph{every} certificate assignment
      satisfying our constraints results in an orientation with a
      unique sink (even if they are not acyclic).
    This highlights the fact that with our orientations, one cannot
    use the well-known trick of constructing a rooted spanning tree and
    electing the root. In the particular subcase of chordal graphs, our
    certificates and constraints additionally ensure an acyclic orientation. For the
    case of dismantlable graphs, we also provide a labeling scheme for
    spanning trees, assuming a distinct node marked as the root
    is given as input. Any certificate assignment, such as the one we give in Section \ref{sec:elec-dism}, satisfying the unique-sink constraints defines a unique leader, which serves as the root.

    Additionally, we show that given a proof-labeling scheme, a silent self-stabilizing
    algorithm can always be constructed with one additional bit of memory
    per node, assuming a Gouda fair scheduler. Although
    \cite{goos2016locally} already mentions that certificates and locally checkable constraints can be interpreted as nondeterministic local
    algorithms, here we formalize this intuition and extend it to defining
    a self-stabilizing algorithm. The one additional bit needed by
    Algorithm~\ref{algo:gcssa} to transform any proof-labeling scheme with $k$-bit certificates into
    a self-stabilizing algorithm is less than prior general
    transformations, such as the one in~\cite{BlinFP14}, which requires
    $O(k + \log n)$ bits, where $n$ is the number of computational
    entities in the system, but works even under an unfair scheduler. So,
    Algorithm~\ref{algo:gcssa} can be used in systems with constant memory
    where~\cite{BlinFP14} cannot be employed. In our case, this means that
    Algorithm \ref{algo:gcssa} produces an edge orientation with a unique
    sink, for any initially arbitrarily directed (or undirected)
    dismantlable or chordal graph.

    \section{Model and Preliminaries}\label{sec:model}
    Let $G = (V,E,\delta)$ be a graph with vertex set $V$, edge set $E$
    and a port labeling function $\delta$. For any vertex $v \in V$, let
    $N(v)$ denote its neighborhood and let $N[v] = N(v) \cup \{v\}$ denote
    its \emph{closed} neighborhood. The degree of a vertex $v$ is the size
    of $N(v)$ and is denoted by $\deg(v)$. The 1-ball at $v$, denoted by
    $B_1(v)$, is the graph induced by $N[v]$. When $G$ is oriented, we
    write $N^+(v)$ (resp. $N^-(v)$) for the outgoing (resp. incoming)
    neighborhood of $v$ in $G$ and refer to nodes in $N^+(v)$ as
    \emph{outneighbors} of $v$. We also let $N^+[v] = N^+(v) \cup \{v\}$
    and $N^-[v] = N^-(v) \cup \{v\}$. Given a graph $G=(V,E)$ and a
      subset $W$ of $V$, the \emph{induced subgraph} $G[W]$ of $G$ induced
      by $W$ is the graph whose vertex set is $W$ and whose edge set
      consists of all of the edges in $E$ that have both endpoints in
      $W$.  We define the port labeling function to be
    $\delta = \{\delta_v\}_{v \in V}$, where
    $\delta_v : N(v) \rightarrow \{1,2,\ldots,\deg(v)\}$ assigns a unique
    port number to each edge incident to $v$ for every $v \in V$ (i.e.,
    $\delta_v$ is a bijection).
    In a directed graph, we say that a triangle $uvw$ is a \emph{directed
      triangle} if the edges of this triangle are directed as
    $\overrightarrow{uv}, \overrightarrow{vw},\overrightarrow{wu}$.

    Let $G = (V,E,\delta)$ be a graph and $S$ be an arbitrary set of states. For any $G$ and any mapping $x_S : V \rightarrow S$, we say that $(G, x_S)$ is a configuration of $G$. We say that $x_S(v)$ is the state from $S$ assigned to $v \in V$ by $x_S$. We define a \emph{task} $T = (O,P)$ as a set of states $O$ and a \emph{problem specification} $P$, where $P$ is a set of configurations over $O$ that we call the \emph{correct} configurations for $T$. It is possible to have several correct configurations in $P$ over the same graph. The notion of a \emph{correct configuration} is specific to the task. In this paper, we consider two tasks: leader certification (\Cref{sec:warm-up} and \Cref{sec:k4-free}) and spanning tree certification (\Cref{sec:k4-free}). We give the definition of correct configurations in the problem specification $P$ and the set of states $O$ for each of these tasks here. We write $x_O$ to denote the mapping $x_O : V \rightarrow O$.
    
    \begin{description}
        \item[Leader Certification] Let $O = \{\textsf{L}, \textsf{NL}\}$ where \textsf{L} denotes the \emph{leader} state and \textsf{NL} denotes the \emph{non-leader} state. The problem specification $P$ for leader election is then expressed as any configuration $(G,x_O)$ for which: $
        \exists! \text{ } v \in V \text{ such that } x_O(v) = \textsf{L} \text{ and } \forall v' \in V \text{ such that } v' \neq v, \text{ } x_O(v') = \textsf{NL}$.
      \item[Spanning Tree Certification] Let $O = \mathbb{N}$. Let $\parent_G(v)$
        be the unique vertex $w \in N(v)$ such that $\delta_v(w) = x_O(v)$
        if $1 \leq x_O(v) \leq \deg(v)$ and let
        $\parent_G(v) = \bot$ otherwise (e.g., if $x_O(v) = 0$). The
        set of correct configurations $P$ is the set of configurations
        $(G,x_O)$ such that:
        \begin{itemize}
        \item $\forall v \in V$, $0 \leq x_O(v) \leq \deg(v)$,
        \item $\exists! \text{ } r \in V$ such that $\parent_G(r) = \bot$
          (i.e., such that $x_O(r) = 0$),
        \item The set of directed edges
          $\{\overrightarrow{vw} \mid v \neq r \text{ and } w = \parent_G(v)\}$ is a
          directed spanning tree.
        \end{itemize} 
    \end{description}

    We assume that the nodes of $G$ do not have identifiers, that is, $G$ is \emph{anonymous}. Let $I$ be a possibly empty set of inputs and call the mapping from $V$ to $I$ the \emph{input function}. Let $L$ be a non-empty set of states, called the \emph{labels}. A \emph{labeling} on $G$ is a mapping $x_L : V \rightarrow L$. For a given labeling $x_L$ and a node $v \in V$, the \emph{view} of $v$ consists of the 1-ball $B_1(v)$, the restriction of the port numbering function $\delta$ and of the labeling $x_L$ to $B_1(v)$ and the input of $v$ (if defined). Note that $B_1(v)$ contains the edges between the neighbors of $v$ and thus the view of $v$ fits within the usual radius 2 verification rules (i.e., 2-PLS), whereas in the classical PLS model, a node only sees the labels of its neighbors. An example of the view of a node is given in Figures \ref{fig:view} and \ref{fig:view-b}. 
    
    \begin{figure}[ht]
        \centering
        \begin{minipage}{.23\textwidth}
            \centering
            \includegraphics[scale=.67,page=1]{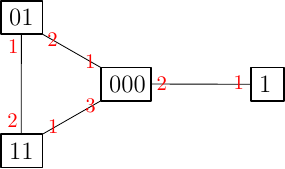}
            \subcaption{\centering}    
            \label{fig:view}
        \end{minipage}
        \begin{minipage}{.2\textwidth}
            \centering
            \includegraphics[scale=.7,page=2]{model.pdf}
            \subcaption{\centering}   
            \label{fig:view-b}
        \end{minipage}
        \begin{minipage}{.23\textwidth}
            \centering
            \includegraphics[scale=.67,page=3]{model.pdf}
            \subcaption{\centering}   
            \label{fig:directed-edges}
        \end{minipage}
        \begin{minipage}{.245\textwidth}
            \centering
            \includegraphics[scale=.7,page=4]{model.pdf}
            \subcaption{\centering} 
            \label{fig:directed-edges-b}
        \end{minipage}
        \caption{(a) A graph with port labels (red) where each node is assigned a label in $\{0,1\}^*$ (black). (b) The view of the node marked in gray. (c) - (d) The $i$-th bit of the label of each node $v$ corresponds to the direction of the edge incident to port $i$ of $v$, for $i \in \{1,\ldots,deg(v)\}$ where 0 denotes an incoming and 1 an outgoing edge. In (c) we show the direction of the edges and in (d) we additionally show the bit of the label of each node corresponding to each edge.}
    \end{figure}
    
    Let $R$ be a set of configurations of the form $(B,x_L)$ where $B$ is a ball of radius 1 and $x_L$ is a labeling defined on the nodes of $B$. We refer to $R$ as a set of \emph{local} configurations. A  configuration is said to be \emph{correct} with respect to $(L, R)$ if
    for every $v \in V$ the pair $(B_1(v), x_L |_{N[v]})$ belongs to $R$, where $x_L |_{N[v]}$ denotes the restriction of $x_L$ to $N[v]$. In other words, the  configuration is correct if the labeling in each ball of radius 1 satisfies the local constraints described by the local configurations in $R$. We refer to the pair $(L,R)$ as a \emph{labeling scheme}. In the terminology of proof-labeling schemes, labels correspond to certificates and $R$ defines the verifier: a node $v$ accepts if $(B_1(v), x_L|_{N[v]}) \in R$.
    
    A labeling scheme $(L,R)$ is of \emph{finite type} if for any
    $1$-ball $B$, there is a finite number of labelings of $B$ in $R$,
    i.e., the set $\{x \in L^{V(B)} \mid (B,x) \in R\}$ is finite. For
    example, if $L$ is finite then any labeling scheme $(L,R)$ is
    of finite type.  Observe that a labeling scheme is of finite type if the
    number of states a node $v$ can take in a correct labeling of its
    $1$-ball is bounded by a function of its degree $\deg(v)$. In the remainder of this paper we only consider labelings of finite type.
    
    A labeling scheme is used to characterize the solutions to a given task $\mathcal{T} = (O,P)$. The key property we prove, is that if a configuration is (locally) correct with respect to the labeling scheme, then it is (globally) correct with respect to the task specification (i.e., it is a correct configuration for the task). Formally, $(L,R)$ is a labeling scheme for a task $(O,P)$ in a family of graphs (that might initially have inputs) $\mathcal{F}$ if there exists $out : L \rightarrow O$ such that the following properties are satisfied\footnote{The symbol $\circ$ denotes the composition operator.}:
    
    \begin{description}
    \item[Existence]\hspace{.04cm} $\forall G \in \mathcal{F}, \text{ } \exists (G,x_L) \text{ s.t. } \forall v \in V, \text{ } (B_1(v), x_L |_{N[v]}) \in R$
    \item[Soundness] $\forall (G,x_L) \text{ s.t. } G \in \mathcal{F}, \text{ if } \forall v \text{ } (B_1(v), x_L |_{N[v]}) \in R \text{ then } (G, \emph{out} \circ x_L) \in P$ 
    \end{description}
    
    We say that $(L,R)$ is a \emph{proof-labeling scheme} (PLS) for $(O,P)$ in $\mathcal{F}$ if it additionally satisfies the following property:
    
    \begin{description}
    \item[Completeness] $\forall (G,x_O) \in P \text{ s.t. } G \in \mathcal{F}, \text{ } \exists (G,x_L) \text{ s.t. } \emph{out} \circ x_L = x_O \text{ and }$
    
    \qquad\qquad\quad$\>\>\>\>\forall v \in V, \text{ } (B_1(v), x_L |_{N[v]}) \in R$
    \end{description}
    
    Completeness corresponds to the usual completeness requirement of proof-labeling schemes, that is, it says that every correct configuration can be certified. The labeling schemes of \Cref{sec:elec-dism,sec:st-dism} satisfy Existence and Soundness, which is sufficient for the transformation of \Cref{algo:gcssa}. In \Cref{sec:completeness-dism}, we combine them to obtain a proof labeling scheme for leader certification that also satisfies Completeness.

    In most of our labeling schemes, the labels correspond to an edge orientation. For these labels, we denote the configuration as $\vec{G}$ and we adopt the following encoding for the labels: the label assigned to each node $v \in V$ is composed of $\deg(v)$ bits where each bit encodes the direction of the corresponding incident edge. We assume that 0 denotes an incoming edge to $v$ and 1 denotes an outgoing edge from $v$. The bit at position $i \in \{1,\ldots,\deg(v)\}$ of the label corresponds to the direction of the edge at port $i$ of $v$. A visualization of this for the graph in \Cref{fig:view} is given in Figures \ref{fig:directed-edges} and \ref{fig:directed-edges-b}.
    For each edge $uv$ where $u$ (resp. $v$) is connected to $v$ (resp. $u$) through port $i$ (resp. $j$), we say that the edge is oriented if the $i$-th bit of $u$ differs from the $j$-th bit of $v$.

    \section{Warm Up: Labelings for Trees and Chordal Graphs}\label{sec:warm-up}
  
    \subsection{Unique Sink in Trees}\label{sec:trees}
    We give a simple labeling that guarantees a unique sink in the class
    of trees. While not new (e.g., one could easily construct a
    certificate following the self--stabilizing leader election algorithm
    of \cite{antonoiu1996trees}), we believe it gives a good intuition
    behind the labelings and proofs in \Cref{sec:chordal,sec:k4-free}. The label
    of each node consists of the orientation of edges incident to the
    node. A  configuration is locally correct in a tree $T$ if the following constraints are satisfied at every node. 
    
    \begin{enumerate}[T1]
        \item\label{rule:tree-all-edges-directed} Each edge is directed.
        \item\label{rule:tree-one-neighbor} For each vertex $v$, $v$ has at  most one outgoing edge.
    \end{enumerate}
    
    An orientation $\vec{T}$ satisfying $T_1$ and $T_2$  exists (i.e.,  the \textsf{Existence} property holds): first direct each edge incident to a leaf toward its unique neighbor, then iteratively orient edges away from leaves in the subgraph of remaining unoriented edges.
    We now show that every locally correct configuration $\vec{T}$ has a unique sink which we
    define to be the leader.

    \begin{theorem}[Soundness]\label{th:tree-certificate-correct}  
        In a tree $T$, 
        any configuration $\vec{T}$ satisfying T\ref{rule:tree-all-edges-directed} and T\ref{rule:tree-one-neighbor} contains a unique sink.
    \end{theorem}
    \begin{normalproof}
        We prove the theorem by induction on the number of vertices of $T$. The base case of a single vertex is trivial. Suppose $T$ contains at least two vertices and let $v$ be a leaf of $T$. Let $\vec{T}$ be a configuration of $T$ and assume that $\vec{T}$ satisfies T\ref{rule:tree-all-edges-directed} and T\ref{rule:tree-one-neighbor}. Observe that the configuration $\vec{T} \setminus \{v\}$ of $T \setminus \{v\}$ also satisfies T\ref{rule:tree-all-edges-directed} and T\ref{rule:tree-one-neighbor}. By the induction hypothesis $\vec{T} \setminus \{v\}$ has a unique sink $t$. We show that $\vec{T}$ also has a unique sink.
    
        Call $u$ the neighbor of $v$ in $T$. Assume $vu$ is oriented $\overrightarrow{uv}$. Then from T\ref{rule:tree-one-neighbor}, $u$ has no other outneighbor in $\vec{T}$, so $u$ is the unique sink in $\vec{T} \setminus \{v\}$. Since $u$ has an outgoing edge in $\vec{T}$, it is not a sink, but $v$ is. Thus, $v$ is the unique sink in $\vec{T}$. Assume now $vu$ is oriented $\overrightarrow{vu}$. 
        Then $t$ remains a sink in $\vec{T}$ and $v$ is not a sink, so $t$ is the unique sink. Hence, in all cases $\vec{T}$ has a unique sink, completing the proof.
    \end{normalproof}
    
    \subsection{Unique Sink and Acyclic Orientation in Chordal Graphs}\label{sec:chordal}
    We generalize the labeling scheme from Section~\ref{sec:trees} to a
    set of constraints guaranteeing acyclicity and a unique sink in chordal
    graphs.  A graph $G=(V,E)$ is chordal if every induced cycle is
      of length three. Every chordal graph has a simplicial vertex $v$,
      that is, a vertex $v$ such that $N(v)$ induces a complete graph. As
      the class of chordal graphs is closed under induced subgraphs
      , any
      chordal graph admits a \emph{simplicial order}, i.e., an ordering
      $v_1, \ldots, v_n$ of $V$ such that for each $i$, $v_i$ is a
      simplicial vertex in $G[\{v_1,\ldots,v_i\}]$. In our labeling
      scheme, each node's label is the orientation of its incident
      edges. We refer to the nodes of $G$ and the labeling as a
      configuration $\vec{G}$.  We define a configuration to be locally correct
     if for every vertex $v$ of a chordal graph, the
    following local constraints hold:
    
    \begin{enumerate}[C1]
        \item\label{rule:chordal-all-edges-directed} Each edge incident to $v$ is directed.
        \item\label{rule:chordal-clique-outneighbors} All outneighbors of $v$ are adjacent, that is, $N^+(v)$ induces a complete graph.
        \item\label{rule:chordal-no-cyclic-triangles} Vertex $v$ does not belong to any directed triangle.
    \end{enumerate}
    
    Observe that checking C\ref{rule:chordal-no-cyclic-triangles} is possible, since the \emph{view} by definition contains the orientation of edges between neighboring nodes. 

      \begin{observation}[Existence]\label{obs:chordal-labeling-exists}
        Any chordal graph $G =(V,E)$ admits an acyclic orientation $\vec{G}$ satisfying C\ref{rule:chordal-all-edges-directed}--C\ref{rule:chordal-no-cyclic-triangles}.
      \end{observation}
      \begin{claimproof}
        Consider a simplicial order $v_1,v_2,\ldots,v_n$ of 
        $G$. Orient each edge $v_iv_j \in E$ as
        $\overrightarrow{v_iv_j}$ if and only if $i > j$ in the simplicial
        order.  This orientation is acyclic and satisfies
        C\ref{rule:chordal-all-edges-directed} and
        C\ref{rule:chordal-no-cyclic-triangles}. For
        C\ref{rule:chordal-clique-outneighbors}, note that for any $i$,
        $N^+(v_i) = N(v_i) \cap \{v_j: j <i\}$. Since $v_i$ is a
        simplicial vertex in $G[\{v_1, \ldots, v_i\}]$, $N^+(v_i)$ induces
        a complete graph and C\ref{rule:chordal-clique-outneighbors} holds
        at $v_i$.
      \end{claimproof}

    \begin{lemma}\label{lem:chordal-acyclic}    
      In a chordal graph $G$, any configuration $\vec{G}$ satisfying
      C\ref{rule:chordal-all-edges-directed}--C\ref{rule:chordal-no-cyclic-triangles}
      is an acyclic orientation of $G$.
    \end{lemma}
    
    \begin{normalproof}    
      Suppose the lemma statement does not hold and $\vec{G}$ contains a
      directed cycle. Take $\mathsf{C} = (u_0, u_1, \ldots, u_{k-1})$ to
      be the smallest such cycle. From
      C\ref{rule:chordal-no-cyclic-triangles}, $k > 3$. Then since $G$
      is chordal, $\mathsf{C}$ has a chord, and without loss of
      generality, $u_1$ is incident to at least three edges:
      $\overrightarrow{u_0u_1}$, $\overrightarrow{u_1u_2}$ and some
      directed edge $u_1u_i$ for $3 \leq i < k$. Depending on the
      orientation of $u_1u_i$, $\vec{G}$ contains either the directed
      cycle $\mathsf{C}' = (u_1, u_i, u_{i+1}, \ldots, u_{k-1}, u_0)$ or
      the directed cycle $\mathsf{C}'' = (u_1, u_2, \ldots, u_i)$. In
      either case, a smaller cycle than $\mathsf{C}$ exists in $G$,
      contradiction.
    \end{normalproof}
    
    \begin{lemma}\label{lem:chordal-unique-sink}
      In a chordal graph $G$, any configuration $\vec{G}$ satisfying
      C\ref{rule:chordal-all-edges-directed}--C\ref{rule:chordal-no-cyclic-triangles}
      contains a unique sink.
    \end{lemma}
    
    \begin{normalproof}  
       We prove the lemma by induction on the number of vertices of $G$.    
      The base case holds for $G$ having one vertex. For $G$ with at least two vertices, let $v$ be a simplicial vertex, that is, $N(v)$ induces a complete subgraph. Consider a configuration $\vec{G}$ of $G$ satisfying C\ref{rule:chordal-all-edges-directed}--C\ref{rule:chordal-no-cyclic-triangles}. Removing $v$ from $\vec{G}$ gives a configuration of $G \setminus \{v\}$ that also satisfies C\ref{rule:chordal-all-edges-directed}--C\ref{rule:chordal-no-cyclic-triangles}, so by the induction hypothesis $\vec{G} \setminus \{v\}$ has a unique sink $t$.
          
            Suppose first that $v$ has an outneighbor $u$ in $\vec{G}$. We claim that in this case, $t$ is also the unique sink of $\vec{G}$. Due to $\vec{G} \setminus \{v\}$ having a unique sink, no vertex other than $t$ can be a sink in $\vec{G}$. If $t$ is not a sink in $\vec{G}$, then $\overrightarrow{tv} \in \vec{G}$, as all remaining edges incident to $t$ are incoming since $t$ is a sink in $\vec{G} \setminus \{v\}$. Since $u,t \in N(v)$ and $v$ is a simplicial vertex, $u \in N^-(t)$. But then $tvu$ is  a directed triangle of $\vec{G}$, contradicting C\ref{rule:chordal-no-cyclic-triangles}. Therefore, $t$ is the unique sink of of $\vec{G}$.
    
            Suppose now that $v$ has no outneighbor, (i.e., $v$ is a sink of $\vec{G}$). For any neighbor $u$ of $v$, any outneighbor $w$ of $u$ in $\vec{G} \setminus \{v\}$ must be a neighbor of $v$ in $\vec{G}$ from C\ref{rule:chordal-clique-outneighbors}. Consequently, $N^+(u) \subseteq N[v]$. We claim that there exists a vertex $t' \in N(v)$ that has no outneighbor in $N(v)$. Suppose otherwise. Then, every $u \in N(v)$ has an outneighbor in $N(v)$ and $\vec{G} \setminus \{v\}$ contains a cycle, which is a contradiction by Lemma \ref{lem:chordal-acyclic}. Since $N(v)$ is a complete graph, this $t'$ is unique. Furthermore, as we have already shown that $t'$ has no outneighbor outside of $N[v]$, $t'$ is the unique sink $t$ of $\vec{G} \setminus \{v\}$. Since $t'=t \in N^-(v)$, $t$ is not a sink in $\vec{G}$. Therefore, $v$ is the unique sink of $\vec{G}$.
        \end{normalproof}
    
        Theorem \ref{th:chordal-certificate-correct} directly follows from 
        \Cref{lem:chordal-acyclic,lem:chordal-unique-sink}.
    
    \begin{theorem}[Soundness]\label{th:chordal-certificate-correct}    
        In a chordal graph $G$, any configuration $\vec{G}$ satisfying C\ref{rule:chordal-all-edges-directed}--C\ref{rule:chordal-no-cyclic-triangles} is an acyclic orientation of $G$ containing a unique sink.
    \end{theorem}
    
    Observe that by \Cref{th:chordal-certificate-correct}, in order to construct a spanning tree in a chordal graph
    satisfying
    C\ref{rule:chordal-all-edges-directed}--C\ref{rule:chordal-no-cyclic-triangles},
    it is enough that each  node different from the sink picks an outgoing edge.

\section{Dismantlable graphs}\label{sec:k4-free}
This section is divided into three parts: a labeling scheme ensuring a
unique sink in dismantlable graphs, a labeling scheme for a rooted
spanning tree, given a unique root, in dismantlable graphs and the combination of the former two labeling schemes to ensure that any node in the graph can be elected.  Let
$G=(V,E)$ be a graph and $u,v$ be two vertices in $V$. We say that $v$
is \emph{dominated} by $u$ in $G$ if $N[v] \subseteq N[u]$. A graph is
dismantlable if there exists an ordering $v_1, \ldots, v_n$ of $V$
such that for each $i > 1$, in $G[\{v_1, \ldots, v_i\}]$, $v_i$ is
dominated by some $v_j$ with $j < i$. So for each $i > 1$, there
exists $j < i$ such that
$N[v_i] \cap \{v_1, \ldots, v_i\} \subseteq N[v_j]$.
The ordering $v_1, \ldots, v_n$ is called a \emph{dismantling order} of $G$. 
We give an
example of a dismantlable graph $\exampleG$ in
\Cref{fig:dismantlable}. Observe that the only dominated vertex in
$\exampleG$ is $a$, which is dominated by $f$. In
$\exampleG\setminus\{a\}$, $b$ (resp. $c$) is dominated by $g$
(resp. $h$). In $\exampleG\setminus\{a,b,c\}$, $d$ (resp.  $e$) is
dominated by $i$. Then in $\exampleG\setminus\{a,b,c,d,e\}$, $f$ is
dominated by $g$; in $\exampleG\setminus\{a,b,c,d,e,f\}$, $g$ is
dominated by $h$; and finally, in
$\exampleG\setminus\{a,b,c,d,e,f,g\}$, $h$ is dominated by $i$. This
gives the dismantling order $(i,h,g,f,e,d,c,b,a)$. Each of the
``dominated by'' relations is depicted by a bold black directed edge
in \Cref{fig:dismantling-order}. Additionally,
\Cref{fig:dismantling-order} contains a directed edge from $v_i$ to
$v_j$ when $v_i$ appears later in the dismantling order than $v_j$.

\subsection{Leader election in  dismantlable graphs}\label{sec:elec-dism}
As for trees and chordal graphs (see Sections \ref{sec:trees} and \ref{sec:chordal}), the label at each node consists of an
orientation on the edges incident to the node and we write $\vec{G}$
to denote a configuration. We define a locally correct configuration
to be one that satisfies
D\ref{rule:all-edges-directed}--D\ref{rule:no-cyclic-triangles} for
every node $v$ of $G$. We show that if all rules are satisfied, $G$
contains a unique sink.
\begin{enumerate}[D1]
    \item\label{rule:all-edges-directed} Each edge incident to $v$ is directed.
    \item\label{rule:dominating-neighbor} Either $v$ has no outgoing edges, or there exists an outneighbor $v^+$ of $v$  adjacent to all other outneighbors of $v$, i.e,  $N^+(v) \subseteq N[v^+]$.
    \item\label{rule:no-cyclic-triangles} Vertex $v$ does not belong to any directed triangle.
\end{enumerate}

Examples of two orientations of the graph $\exampleG$ of
\Cref{fig:dismantlable} satisfying
D\ref{rule:all-edges-directed}--D\ref{rule:no-cyclic-triangles} are
given in Figures~\ref{fig:dismantling-order}
and~\ref{fig:not-acyclic}. In the next lemma, we show that any
dismantling order of the vertices of a graph $G$ yields an
acyclic orientation of $G$ satisfying
D\ref{rule:all-edges-directed}--D\ref{rule:no-cyclic-triangles}.  The
orientation in \Cref{fig:dismantling-order} is precisely the
orientation derived from the dismantling order $(i,h,g,f,e,d,c,b,a)$
of the graph $\exampleG$.

\begin{observation}[Existence]\label{obs:dismantlable-labeling-exists}
    Any dismantlable graph $G=(V,E)$ admits an acyclic orientation $\vec{G}$ satisfying D\ref{rule:all-edges-directed}--D\ref{rule:no-cyclic-triangles}.
\end{observation}
\begin{claimproof}
  Let $v_1,v_2,\ldots,v_n$ be the dismantling order of the nodes of
  some dismantlable graph $G$. Orient each edge $v_iv_j \in E$ as $\overrightarrow{v_iv_j}$ if and only if $i > j$
  in the dismantling order.  This orientation is acyclic and satisfies
  D\ref{rule:all-edges-directed} and D\ref{rule:no-cyclic-triangles}.
  For D\ref{rule:dominating-neighbor}, consider any node $v_i$. If
  $i = 1$, then $N^+(v) = \emptyset$, and
  D\ref{rule:dominating-neighbor} satisfied for $v$. Suppose now
  that $i > 1$. Since $v_1, v_2, \ldots, v_n$ is a dismantling order,
  there exists $j < i$ such that $N[v_i] \cap \{v_1, \ldots v_i\} \subseteq N[v_j]$. By the
  definition of $\vec{G}$,
  $N^+(v_i) = N(v_i) \cap \{v_1, \ldots v_i\}$ and thus
  $N^+(v_i) \subseteq N[v_j]$, i.e., D\ref{rule:dominating-neighbor}
  holds for $v_i$ with $v_i^+ = v_j$.
\end{claimproof}

Notice that in the case of dismantlable graphs, one cannot hope to
have a property similar to \Cref{lem:chordal-acyclic}. As
in~\cite{chalopin2024stabilising}, one can construct configurations
satisfying
D\ref{rule:all-edges-directed}-D\ref{rule:no-cyclic-triangles} that
contain directed cycles (e.g., as in \Cref{fig:not-acyclic}) and are
thus not obtained from a dismantling order.
\begin{figure}[h]
    \centering
    \begin{minipage}[t]{.31\textwidth}
        \centering
        \includegraphics[scale=.4,page=7]{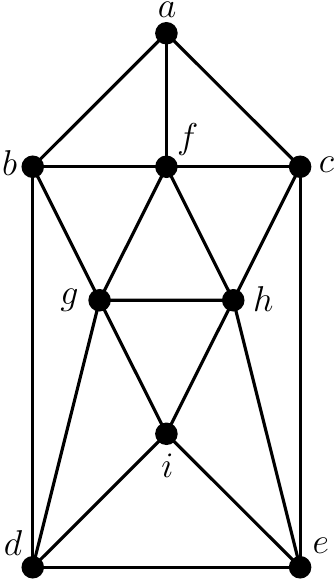}
        \subcaption{\centering}
        \label{fig:dismantlable}
    \end{minipage}
    \begin{minipage}[t]{.31\textwidth}
        \centering
        \includegraphics[scale=.4,page=11]{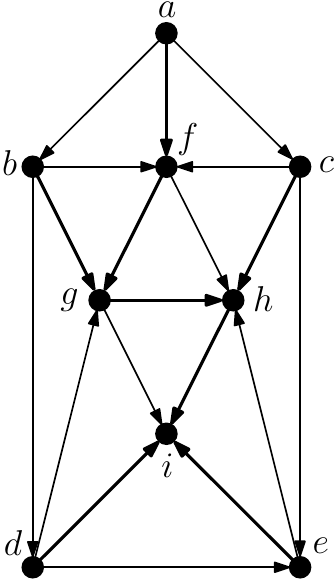}
        \subcaption{\centering}
        \label{fig:dismantling-order}
    \end{minipage}
    \begin{minipage}[t]{.31\textwidth}
        \centering
        \includegraphics[scale=.4,page=13]{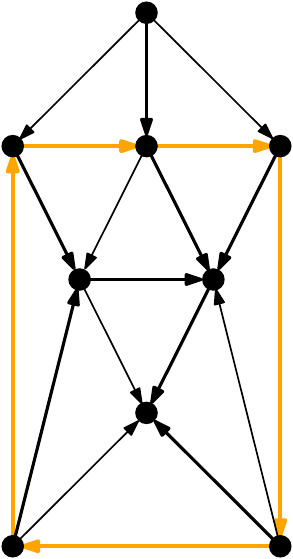}
        \subcaption{\centering}
        \label{fig:not-acyclic}
    \end{minipage}
    \caption{(a) An example of a dismantlable graph $\exampleG$. (b)
      An acyclic orientation of $\exampleG$ that satisfies
      D\ref{rule:all-edges-directed}--D\ref{rule:no-cyclic-triangles}. (c)
      An orientation satisfying
      D\ref{rule:all-edges-directed}--D\ref{rule:no-cyclic-triangles}
      that contains a cycle (denoted by the orange edges).}
\end{figure}

\begin{theorem}[Soundness]\label{th:certificate-K4-free}  
  Let $G$ be a dismantlable graph. If a configuration $\vec{G}$
  satisfies
  D\ref{rule:all-edges-directed}--D\ref{rule:no-cyclic-triangles},
  then $\vec{G}$ contains a unique sink.
\end{theorem}

\begin{normalproof}
  Consider a graph $G$ and a configuration $\vec{G}$ of $G$ satisfying
  constraints
  D\ref{rule:all-edges-directed}--D\ref{rule:no-cyclic-triangles}. We
  prove the theorem by induction on the number of vertices in $G$. The
  theorem trivially holds if $G$ has one vertex. Suppose $G$ has at
  least two vertices. Let $v$ be a vertex dominated by some vertex $u$
  in $G$. 

  First, in the next lemma, we construct a simple directed path
  $(u_0=u, u_1, \ldots, u_k)$ in the neighborhood of $v$ starting at
  $u$ and enjoying interesting adjacency, orientation and domination
  properties that are fundamental in our proof. Intuitively, the path
  starts at a node $u = u_0$ dominating $v$, and we let
  $u_{i+1} = u_i^+$ until we reach a vertex $u_i$ such that either
  $u_i$ is an outneighbor of $v$, or such that we can set $v = u_i^+$.
  We show that all the vertices of this path are distinct neighbors of
  $v$ and that consequently, we eventually reach a vertex $u_i$ such
  that $u_i \in N^+(v)$ or $N^+(u_i) \subseteq N[v]$. An illustration
  of this construction is given in Figure~\ref{fig:Puk}.
  \begin{figure}[ht]
    \centering
    \hfill \includegraphics[scale=.4,page=2]{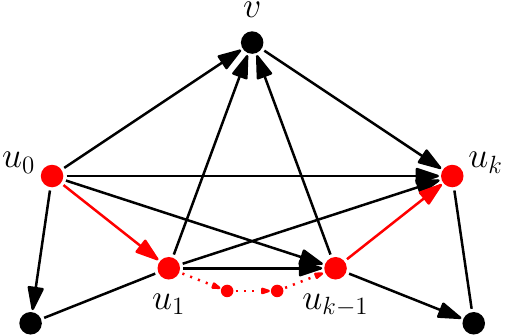}\hfill%
    \includegraphics[scale=.4,page=3]{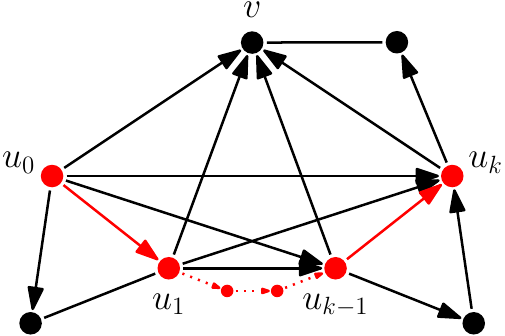}\hfill \     
    \caption{In red, the path $(u_0, u_1, \ldots, u_k)$ constructed in
      Lemma~\ref{lemma:construction-Puk}. For property
      (P\ref{claim-ind-7}), either $u_k \in N^+(v)$ (left) or
      $N^+(u_k) \subseteq N[v]$ (right).}
    \label{fig:Puk}
  \end{figure}
  
  \begin{lemma}\label{lemma:construction-Puk}
    There exists a sequence of distinct vertices
    $(u=u_0, u_1, \ldots, u_k)$ satisfying the following properties:
    \begin{enumerate}[{(P}1)]
    \item\label{claim-ind-1} $\forall 0 \leq i <k$, $v \in N^+(u_i)$,
      and
    \item\label{claim-ind-2} $\forall 0 \leq i <k$,
      $N^+(u_i) \not\subset N[v]$, and
    \item\label{claim-ind-3} $\forall 0 \leq i <k$,
      $u_{i+1} \in N^+(u_i)$ and $N^+(u_i) \subseteq N[u_{i+1}]$ (i.e.,
      $u_{i+1} = u_i^+$), and
    \item\label{claim-ind-4} $\forall 0 \leq i <j \leq k$,
      $u_j \in N^+(u_i)$, and
    \item\label{claim-ind-5} $v \in N(u_k)$, and
    \item\label{claim-ind-6} $\forall 0 \leq i \leq k$, if
      $v \in N^+(u_i)$, then $N^+(v) \subseteq N^+(u_i)$, and
    \item\label{claim-ind-7} either $v \in N^-(u_k)$ or
      $N^+(u_k) \subseteq N[v]$.
  \end{enumerate}
  \end{lemma}

  \begin{normalproof}
    We construct iteratively a sequence $(u=u_0, u_1, \ldots, u_k)$ of
    distinct vertices that satisfies properties (P\ref{claim-ind-1}),
    (P\ref{claim-ind-2}), (P\ref{claim-ind-3}), (P\ref{claim-ind-4}),
    (P\ref{claim-ind-5}), and (P\ref{claim-ind-6}) until we reach a
    vertex $u_k$ that satisfies (P\ref{claim-ind-7}).

    For $k = 0$, we let $u_0 = u$.  Observe that (P\ref{claim-ind-1}),
    (P\ref{claim-ind-2}), (P\ref{claim-ind-3}), (P\ref{claim-ind-4})
    hold by vacuity. Note that (P\ref{claim-ind-5}) holds since $u$
    dominates $v$. To show that (P\ref{claim-ind-6}) holds, suppose by
    contradiction that $v \in N^+(u_0)$ and that there exists
    $w \in  N^+(v) \setminus N^+(u_0)$. Since $u$
    dominates $v$, $w \in N[u_0]$ and $w \neq u_0$ as
    $u_0 \in N^-(v)$. Consequently, $w \in N^-(u_0)$. But then, we
    obtain a contradiction with D\ref{rule:no-cyclic-triangles}
    applied to the triangle $u_{0}vw$. This shows that
    (P\ref{claim-ind-6}) holds.

    Suppose now that we have constructed a sequence
    $(u_0, u_1, \ldots, u_k)$ of distinct vertices such that
    (P\ref{claim-ind-1}), (P\ref{claim-ind-2}), (P\ref{claim-ind-3}),
    (P\ref{claim-ind-4}), (P\ref{claim-ind-5}), and
    (P\ref{claim-ind-6}) hold. If $v \in N^-(u_k)$ or
    $N^+(u_k) \subseteq N[v]$, then the sequence
    $(u_0, u_1, \ldots, u_k)$ satisfies Properties
    (P\ref{claim-ind-1}), (P\ref{claim-ind-2}), (P\ref{claim-ind-3}),
    (P\ref{claim-ind-4}), (P\ref{claim-ind-5}), (P\ref{claim-ind-6}),
    and (P\ref{claim-ind-7}), and the lemma is proven.

    Suppose now that $v \in N^+(u_k)$ and $N^+(u_k) \not\subset
    N[v]$. Since $N^+(u_k) \neq \emptyset$, by
    D\ref{rule:dominating-neighbor}, there exists $u_k^+ \in N^+(u_k)$
    such that $N^+(u_k) \subseteq N[u_k^+]$. Let $u_{k+1} = u_k^+$ and
    note that $u_{k+1} \neq v$ since $N^+(u_k) \not\subset N[v]$ while
    $N^+(u_k) \subseteq N[u_k^+]$. We now show that
    $(u_0, u_1, \ldots, u_k, u_{k+1})$ is a sequence of distinct
    vertices satisfying (P\ref{claim-ind-1}), (P\ref{claim-ind-2}),
    (P\ref{claim-ind-3}), (P\ref{claim-ind-4}), (P\ref{claim-ind-5}),
    and (P\ref{claim-ind-6}).

    First observe that since $u_k \in N^-(u_{k+1})$ and
    $u_k \in N^+(u_i)$ for any $0 \leq i < k$ by (P\ref{claim-ind-4}),
    $u_{k+1}$ is necessarily distinct from $u_i$ for any
    $0 \leq i \leq k$. This shows that
    $(u_0, u_1, \ldots, u_k, u_{k+1})$ is a sequence of distinct
    vertices.  Since we have assumed that $v \in N^+(u_k)$ and
    $N^+(u_k) \not\subset N[v]$, (P\ref{claim-ind-1}) and
    (P\ref{claim-ind-2}) are satisfied. By the definition of
    $u_{k+1}$, (P\ref{claim-ind-3}) also holds, and since
    $v \in N^+(u_k) \subseteq N[u_{k+1}]$, (P\ref{claim-ind-5}) is
    also satisfied.

    We now consider (P\ref{claim-ind-4}). First, observe that $u_{k+1} \in N^+(u_k)$, by the
    definition of $u_{k+1}$, and thus we only need to show that
    $u_{k+1} \in N^+(u_i)$ for any $0 \leq i < k$.  Note that for any
    $0 \leq i < k$, if $u_{k+1} \in N(u_i)$, then
    $u_{k+1} \in N^+(u_i)$ by D\ref{rule:no-cyclic-triangles} applied
    to the triangle $u_iu_ku_{k+1}$, as $u_k \in N^+(u_i)$ and
    $u_{k+1} \in N^+(u_k)$. Consequently, for any $0 \leq i < k$,
    $u_{k+1} \in N(u_i)$ if and only if $u_{k+1} \in
    N^+(u_i)$. To establish (P\ref{claim-ind-4}), we thus show
    by induction on $i$ that for any $0 \leq i < k$,
    $u_{k+1} \in N(u_i)$. Since we have established that
    (P\ref{claim-ind-5}) holds,
    $u_{k+1} \in N(v) \subseteq N[u] = N[u_0]$ and thus
    $u_{k+1} \in N(u_0)$ as $u_{k+1} \neq u_0$. Suppose now that
    $u_{k+1} \in N(u_i)$ for $0 \leq i < k-1$. Since this implies that
    $u_{k+1} \in N^+(u_i)$ and since $N^+(u_i) \subseteq N[u_{i+1}]$
    by (P\ref{claim-ind-3}), we get $u_{k+1} \in N(u_{i+1})$ as
    $u_k \neq u_{i+1}$. This shows that (P\ref{claim-ind-4}) holds.

    We now show that (P\ref{claim-ind-6}) holds for $u_{k+1}$. By
    contradiction, assume that $v \in N^+(u_{k+1})$ and that there
    exists $w \in N^+(v) \setminus N^+(u_{k+1})$.
    Since $v \in N^+(u_k)$ and since (P\ref{claim-ind-6}) holds for
    $u_k$, $w \in N^+(v) \subseteq N^+(u_k) \subseteq N[u_{k+1}]$
    where the last inclusion holds by (P\ref{claim-ind-3}). Since
    $u_{k+1} \in N^-(v)$, $w \neq u_{k+1}$ and since
    $w \notin N^+(u_{k+1})$, we have $w \in N^-(u_{k+1})$. But then,
    we obtain a contradiction with D\ref{rule:no-cyclic-triangles}
    applied to the triangle $u_{k+1}vw$. This shows that
    (P\ref{claim-ind-6}) holds, and thus
    $(u_0, u_1, \ldots, u_k, u_{k+1})$ is a sequence of distinct
    vertices satisfying (P\ref{claim-ind-1}), (P\ref{claim-ind-2}),
    (P\ref{claim-ind-3}), (P\ref{claim-ind-4}), (P\ref{claim-ind-5}),
    and (P\ref{claim-ind-6}).

    Consequently, we can construct a a sequence
    $(u=u_0, u_1, \ldots, u_k)$ of distinct vertices satisfying
    (P\ref{claim-ind-1}), (P\ref{claim-ind-2}), (P\ref{claim-ind-3}),
    (P\ref{claim-ind-4}), (P\ref{claim-ind-5}), and
    (P\ref{claim-ind-6}) until we reach a vertex vertex $u_k$
    satisfying (P\ref{claim-ind-7}). Since $G$ has a finite number of
    vertices, we eventually reach such a vertex.
  \end{normalproof}

  For the remainder of the proof, let $P = (u=u_0, u_1, \ldots, u_k)$
  be the path obtained from \Cref{lemma:construction-Puk} satisfying
  (P\ref{claim-ind-1}), (P\ref{claim-ind-2}), (P\ref{claim-ind-3}),
  (P\ref{claim-ind-4}), (P\ref{claim-ind-5}), (P\ref{claim-ind-6}),
  and (P\ref{claim-ind-7}). We first established some properties
  satisfied by the $u_i$s and their neighbors. In the following lemma,
  we show that if there exists a vertex $w$ such that $v$ is the only
  possible choice for $w^+$ (given by
  D\ref{rule:dominating-neighbor}), then $w$ is an outneighbor of all
  $u_i$s.
 
  \begin{lemma}\label{lemma:orientation-wui}
    For any vertex $w \in N^-(v)$, either there exists
    $w^+ \in N^+(w) \setminus \{v\}$ such that
    $N^+(w) \subseteq N[w^+]$, or $w = u_k$, or $w \in N^+(u_i)$ for
    all $0 \leq i \leq k$.
  \end{lemma}

  \begin{normalproof}
    Consider a vertex $w \in N^-(v)$ such that $N^+(w) \subseteq N[v]$
    and for any $w' \in N^+(w) \setminus \{v\}$,
    $N^+(w) \not\subset N[w']$ (i.e., the vertex $w^+$ given by
    D\ref{rule:dominating-neighbor} is necessarily $v$).  Assume that
    $w \neq u_k$ and observe that by (P\ref{claim-ind-2}), this
    implies that $w \neq u_i$ for any $0 \leq i \leq k$.

    We prove by induction on $i$ that $N^+[w] \subseteq N[u_i]$ for
    any $0 \leq i \leq k$. Observe that since $w \neq u_i$, this
    implies that $w \in N^+(u_i)$ as otherwise, we can set
    $w^+ = u_i$, which contradicts the assumption that necessarily
    $w^+ = v$.  Observe that the inductive property
    holds for $i =0$, since
    $N^+[w] = N^+(w) \cup \{w\} \subseteq N[v] \subseteq N[u_0]$.
    Assume now that $w \in N^+(u_i)$ and $N^+(w) \subseteq N[u_i]$ for
    some $0 \leq i < k$. Note that for any $t \in N^+(w) \subseteq N(u_i)$, by
    D\ref{rule:no-cyclic-triangles} applied to the triangle $u_iwt$,
    we deduce that $t \in N^+(u_i)$. Consequently,
    $N^+[w] \subseteq N^+(u_i) \subseteq N[u_{i+1}]$, where the last
    inclusion holds by (P\ref{claim-ind-3}). This shows that the
    inductive property holds for $i+1$ and this ends the proof of the
    lemma.
  \end{normalproof}

  We show that if $t$ is a sink in $G\setminus \{v\}$ but not in $G$,
  then $t \in N^+[u_k]$.

  \begin{lemma}\label{lemma:disappearing-sink}
    For any vertex $t \in V$, if $N^+(t) = \{v\}$, then $t = u_k$ or
    $t \in N^+(u_k) \cap N^-(v)$.
  \end{lemma}

  \begin{normalproof}
    Note that since $N^+(t) = \{v\}$, we have
    $N^-(t) = N(t)\setminus \{v\}$. Since $u_k \in N^+(u_i)$ for all
    $0 \leq i < k$, we have $t \neq u_i$ for all $0 \leq i < k$.  We
    prove by induction on $i$ that $t \in N[u_i]$ for all
    $0 \leq i \leq k$.  For $i = 0$, this holds since
    $t \in N(v) \subseteq N[u_0]$. Suppose that $t \in N[u_i]$ for
    $0 \leq i < k$, and observe that since $N^+(t) = \{v\}$, we have
    $t \in N^+(u_i) \subseteq N[u_{i+1}]$ where the last inclusion
    holds by (P\ref{claim-ind-3}), establishing the induction
    hypothesis for $i+1$. Consequently, $t \in N[u_k]$ and thus,
    either $t = u_k$, or $t \in N^+(u_k) \cap N^-(v)$.
  \end{normalproof}

  We now show that the inneighbors of $v$ that are outneighbors of
  $u_k$ are also outneighbors of all $u_i$s.
  
  \begin{lemma}\label{lemma:wvmoinsukplus}
    For any vertex $w \in  N^+(u_k) \cap N^-(v)$, we have
    $w \in N^+(u_i)$ for all $0 \leq i \leq k$.
  \end{lemma}

  \begin{normalproof}
    Observe first that by D\ref{rule:no-cyclic-triangles} applied to
    the triangle $u_kwv$, we have $v \in N^+(u_k)$. Note that
    $w \neq u_i$ for all $0 \leq i < k$, as $w \in N^+(u_k)$ while
    $u_i \in N^-(u_k)$ by (P\ref{claim-ind-4}).  Observe also that for
    any $0 \leq i < k$, if $w \in N[u_i]$, then by
    D\ref{rule:no-cyclic-triangles} applied to the triangle $u_iu_kw$
    we get $w \in N^+(u_i)$, as $u_k \in N^+(u_i)\cap N^-(w)$. We show
    by induction on $i$ that $w \in N[u_i]$ for all
    $0 \leq i \leq k-1$. The property holds for $i=0$ since
    $w \in N(v) \subseteq N[u_0]$. Moreover, if $w \in N[u_i]$ for
    $0 \leq i < k-1$, then $w \in N^+(u_i) \subseteq N[u_{i+1}]$,
    where the inclusion holds by (P\ref{claim-ind-3}). This shows that
    the inductive property holds for $i+1$, and this ends the proof of
    the lemma.
  \end{normalproof}

  In order to prove \Cref{th:certificate-K4-free}, we consider
  different cases.  We first distinguish two cases, depending on the
  orientation of the edge $u_kv$.

  \begin{case}\label{case-vu}
    The edge $u_kv$ is oriented as $\overrightarrow{vu_k}$.
  \end{case}

  \begin{claimproof}
    Let $G' = G \setminus \{v\}$ and
    $\vec{G}' = \vec{G} \setminus \{v\}$. Since
    D\ref{rule:all-edges-directed} and D\ref{rule:no-cyclic-triangles}
    hold in $\vec{G}$, they also hold in $\vec{G}'$. To apply the
    induction hypothesis to $\vec{G}'$, it remains to show that
    D\ref{rule:dominating-neighbor} holds in $\vec{G}'$. By
    D\ref{rule:dominating-neighbor} in $\vec{G}$, for any $w$, there
    exists $w^+ \in N^+(w)$ such that $N^+(w) \subseteq N[w^+]$. If
    there exists such a $w^+$ that is different from $v$, then
    D\ref{rule:dominating-neighbor} also holds at $w$ in
    $\vec{G}'$. Suppose now that there exists a vertex $w$ such that
    $v$ is the only vertex in $N^+(w)$ satisfying
    $N^+(w) \subseteq N[v]$. Note that since $v \in N^-(u_k)$,
    $w \neq u_k$, and thus by \Cref{lemma:orientation-wui},
    $w \in N^+(u_k)$. But this implies that the triangle $wvu_k$ does
    not satisfy D\ref{rule:no-cyclic-triangles}, a
    contradiction. Consequently, D\ref{rule:dominating-neighbor} holds
    in $\vec{G}'$.

    Since
    D\ref{rule:all-edges-directed}--D\ref{rule:no-cyclic-triangles}
    are satisfied in $\vec{G}'$, $\vec{G}'$ contains a unique sink
    $t$. Suppose that $t$ is not a sink in $\vec{G}$ and note that
    this implies that $N^+(t) = \{v\}$ in $\vec{G}$. Since
    $v \notin N^+(u_k)$, this implies that $t \neq u_k$, and thus by
    Lemma~\ref{lemma:disappearing-sink}, $v \in N^+(t)$ and
    $t \in N^+(u_k)$. Since $u_k \in N^+(v)$, the triangle $tvu_k$ is
    a directed triangle, contradicting
    D\ref{rule:no-cyclic-triangles}.
    Consequently, $t$ is also a sink in $\vec{G}$. Since all other
    vertices of $\vec{G}'$ have outgoing edges in $\vec{G}'$ (and thus
    in $\vec{G}$), and since $v$ is not a sink in $\vec{G}$ (as
    $u_k \in N^+(v)$), $t$ is the unique sink of $\vec{G}$. This ends
    the proof of \Cref{case-vu}.
  \end{claimproof}

  \begin{case}\label{case-uv}
    The edge $u_kv$ is oriented as  $\overrightarrow{u_kv}$.
  \end{case}
  \begin{claimproof}
    By (P\ref{claim-ind-7}), we know that $N^+(u_k) \subseteq N[v]$.
    We further distinguish two cases, depending on whether
    $N^+(u_k) \cap N^-(v)$ is empty or not.
    
    \begin{subcase}\label{case-uplusvmoinsempty}
      $N^+(u_k) \cap N^-(v) = \emptyset$.
    \end{subcase}
    
    \begin{claimproof}
      We first show that $N^+(u_k) = N^+[v]$. By (P\ref{claim-ind-7})
      and since $v \in N^+(u_k)$, we have $N^+(u_k) \subseteq N[v]$.
      Moreover, since $N^+(u_k)\cap N^-(v) = \emptyset$, this implies
      that $N^+(u_k) \subseteq N^+[v]$. By (P\ref{claim-ind-6}), we
      have $N^+(v) \subseteq N^+(u_k)$. Consequently, we have
      $N^+(u_k) = N^+(v) \cup \{v\}$.
      
      Let $G' = G \setminus \{v\}$ and
      $\vec{G}' = \vec{G} \setminus \{v\}$. Since
      D\ref{rule:all-edges-directed} and
      D\ref{rule:no-cyclic-triangles} hold in $\vec{G}$, they also
      hold in $\vec{G}'$. Moreover, as in \Cref{case-vu}, if
      D\ref{rule:dominating-neighbor} does not hold in $\vec{G}'$,
      then there exists a vertex $w$ such that $v$ is the only vertex
      in $N^+(w)$ satisfying $N^+(w) \subseteq N[v]$. By
      Lemma~\ref{lemma:orientation-wui}, this implies that either
      $w=u_k$ or $w \in N^+(u_k)$. In the second case, this implies
      that $w \in N^-(v) \cap N^+(u_k)$, contradicting our hypothesis.
      We now suppose that $w=u_k$ and we obtain a contradiction by
      showing that D\ref{rule:dominating-neighbor} holds at $u_k$ in
      $\vec{G}'$. Since $N^+(u_k) = N^+[v]$, either
      $N^+(v) = \emptyset$ and $N^+(u_k) = \{v\}$, or there exists
      $v^+ \in N^+(v) \subseteq N^+(u_k)$ such that
      $N^+(v) \subseteq N[v^+]$, implying
      $N^+(u_k) = N^+(v) \cup \{v\} \subseteq N[v^+]$. In either case,
      D\ref{rule:dominating-neighbor} is satisfied for $u_k$ in
      $\vec{G}'$ and thus D\ref{rule:dominating-neighbor} holds in
      $\vec{G}'$.

      By induction hypothesis $\vec{G}'$ contains a unique sink
      $t$. Notice that as $N^+(u_k) = N^+[v]$, $v$ is a sink of
      $\vec{G}$ if and only if $u_k$ is a sink of
      $\vec{G}'$. Therefore, if $t = u_k$, then $t$ is not a sink in
      $\vec{G}$ but $v$ is, and consequently $v$ is the unique sink of
      $\vec{G}$.  Suppose now that $t \neq u_k$.  Since
      $N^+(u_k) \cap N^-(v) = \emptyset$, by the contrapositive of
      Lemma~\ref{lemma:disappearing-sink}, we obtain that $t$ is also
      a sink in $\vec{G}$.  Since $u_k$ is not a sink in $\vec{G'}$,
      $v$ is not a sink in $\vec{G}$ as
      $N^+(v) = N^+(u_k) \setminus \{v\}$ is non-empty. Therefore, $t$
      is the unique sink of $\vec{G}$.

      In any case, $\vec{G}$ contains a unique sink concluding the
      proof of \Cref{case-uplusvmoinsempty}.
    \end{claimproof}
  
    \begin{subcase}\label{case-uplusvmoinsnotempty}
      $N^+(u_k) \cap N^-(v) \neq \emptyset$.
    \end{subcase}
    \begin{claimproof}
      In this case, we consider the configuration $\vec{G}^*$ obtained
      from $\vec{G}$ where for each $w \in N^+(u_k) \cap N^-(v)$ in
      $\vec{G}$, we have replaced the arc $\overrightarrow{u_kw}$ by
      the arc $\overrightarrow{wu_k}$.  We show that
      D\ref{rule:all-edges-directed}--D\ref{rule:no-cyclic-triangles}
      are satisfied in $\vec{G}^*$. Observe that since all edges are
      directed in $\vec{G}$, they are all directed in $\vec{G}^*$ and
      D\ref{rule:all-edges-directed} holds.
        
      Suppose that there exists a directed triangle in
      $\vec{G}^*$. Then necessarily, this triangle contains an edge
      $\overrightarrow{wu_k}$ with $w \in N^+(u_k) \cap N^-(v)$ in
      $\vec{G}$. Let $x$ be the vertex such that
      $\overrightarrow{wu_k}$, $\overrightarrow{u_kx}$,
      $\overrightarrow{xw}$ form a triangle in $\vec{G}^*$. Note that
      $x \neq v$ as $w \in N^-(v) \cap N^+(x)$ in $\vec{G}$ and
      $\vec{G}^*$. Since $\overrightarrow{u_kx}$ is in $\vec{G}^*$,
      $x \in N^+(u_k)$ in $\vec{G}^*$ and in $\vec{G}$ (the set of
      outneighbors of $u_k$ in $\vec{G}^*$ is a subset of its set of
      outneighbors in $\vec{G}$).  Since in $\vec{G}$,
      $N^+(u_k) \subseteq N[v]$ by (P\ref{claim-ind-7}), $x$ is a
      neighbor of $v$. Since $w \in N^+(x)$ and $v \in N^+(w)$, by
      D\ref{rule:no-cyclic-triangles} applied to the triangle $xwv$,
      we have $x \in N^-(v)$ in $\vec{G}$. Consequently, in $\vec{G}$,
      $x \in N^+(u_k) \cap N^-(v)$, and thus in $\vec{G}^*$,
      $x \in N^-(u_k)$, a contradiction. This shows that
      D\ref{rule:no-cyclic-triangles} holds in $\vec{G}^*$.
      
      We now show that D\ref{rule:dominating-neighbor} holds in
      $\vec{G}^*$.  For any
      $x \notin \left(N^+(u_k) \cap N^-(v)\right)\cup \{u_k\}$, $x$ has the same
      set of outneighbors in $\vec{G}$ and in $\vec{G}^*$, and thus
      D\ref{rule:dominating-neighbor} holds for $x$ in $\vec{G}^*$ as
      it holds in $\vec{G}$. Since $N^+(u_k) \subseteq N[v]$ in
      $\vec{G}$ by (P\ref{claim-ind-7}), since $N^+(u_k)$ in
      $\vec{G}^*$ is a subset of $N^+(u_k)$ in $\vec{G}$, and since
      $v \in N^+(u_k)$ in $\vec{G}^*$, D\ref{rule:dominating-neighbor}
      also holds for $u_k$ in $\vec{G}^*$ (with $u_k^+ = v$).

      Consider now $w \in N^+(u_k) \cap N^-(v)$ in $\vec{G}$ and let
      $w^+ \in N^+(w)$ such that $N^+(w) \subseteq N[w^+]$ in
      $\vec{G}$ (note that it is possible to have $w^+ =v$). We show
      that $w^+$ is a neighbor of $u_k$ in $G$ and thus that
      D\ref{rule:dominating-neighbor} is still satisfied at $w$ in
      $\vec{G}^*$ with the same $w^+$.  By
      Lemma~\ref{lemma:wvmoinsukplus}, for all $0 \leq i \leq k$, we
      have $w \in N^+(u_i)$ and thus $w^+ \neq u_i$.  Note that for
      any $0 \leq i \leq k$, if $w^+ \in N[u_i]$, then by
      D\ref{rule:no-cyclic-triangles} applied to the triangle
      $u_iww^+$, we get $w^+ \in N^+(u_i)$.  We show by induction on
      $i$ that $w^+ \in N[u_i]$ for all $0 \leq i \leq k$. Since
      $v \in N^+(w) \subseteq N[w^+]$, the property holds for $i = 0$
      as $w^+ \in N[v] \subseteq N[u_0]$. Suppose now that
      $w^+ \in N[u_i]$ for some $0 \leq i < k$. This implies that
      $w^+ \in N^+(u_i) \subseteq N[u_{i+1}]$ where the inclusion
      holds by (P\ref{claim-ind-3}). Consequently, $w^+$ is a
      neighbor of $u_k$. Therefore since $u_k$ is the only outneighbor
      of $w$ in $\vec{G}^*$ that is not an outneighbor of $w$ in
      $\vec{G}$, we obtain that $N^+(w) \subseteq N[w^+]$ is also true
      in $\vec{G}^*$.  Therefore, D\ref{rule:dominating-neighbor} is
      satisfied at $w$ in $\vec{G}^*$. So
      D\ref{rule:dominating-neighbor} holds in $\vec{G}^*$.
        
      Consequently, by \Cref{case-uplusvmoinsempty} applied to $\vec{G}^*$,
      we deduce that there exists a unique sink $t$ in
      $\vec{G}^*$. Note that
      $t \notin (N^+(u_k) \cap N^-(v))\cup \{u_k\}$ as every vertex
      $w \in (N^+(u_k) \cap N^-(v))\cup \{u_k\}$ has $v$ as an
      outneighbor in $\vec{G}$ and in $\vec{G}^*$. So, $t$ is a sink
      in $\vec{G}$ and since no node of
      $(N^+(u_k) \cap N^-(v))\cup \{u_k\}$ is a sink in $\vec{G}$, $t$ is
      the unique sink of $\vec{G}$.
      
      This ends the proof of \Cref{case-uplusvmoinsnotempty}.
    \end{claimproof}

    This ends the proof of \Cref{case-uv}.
  \end{claimproof}
  
  Theorem \ref{th:certificate-K4-free} follows from the analysis in
  Cases \ref{case-vu} and \ref{case-uv}.
\end{normalproof}

\subsection{Spanning Tree}\label{sec:st-dism}
We present a labeling constructing a spanning tree in any dismantlable graph $G = (V,E)$, with inputs. The input function assigns each $v \in V$ an input in $I = \{0,1\}$, with exactly one node receiving 1 (the root $r$) and all others receiving 0. We define $L = \{0,1,2\}$ to be the set of labels and each node is assigned a value from $L$ by $x_L : V \rightarrow L$. The following must hold for every node $v$ of $G$ with label $l$.

\begin{enumerate}[N1]
    \item\label{rule:at-least-one-smaller} Node $v \neq r$ has at least one neighbor with label $(l - 1) \mod 3$.
    \item\label{rule:bigger-smaller-not-touching} Node $v$ is not in any three-node triangle labeled $\{0,1,2\}$.
\end{enumerate}

We show that if
N\ref{rule:at-least-one-smaller}--N\ref{rule:bigger-smaller-not-touching}
are satisfied for every node in a dismantlable graph $G$, then adding
an edge from any node $v$ distinct from the root with label $l$ to any
neighbor of $v$ with label $l-1 \mod 3$, we obtain a spanning tree.

We first prove the Existence property.

\begin{lemma}[Existence]\label{lem:labeling-exists}
  In any graph $G$ there exists a labeling that satisfies
  N\ref{rule:at-least-one-smaller}--N\ref{rule:bigger-smaller-not-touching}
  for every node in $G$.
\end{lemma}
\begin{normalproof}
  We write $d(v,r)$ to denote the distance between a node $v$ and the
  root $r$. We construct the assignment as follows. For each node $v$
  at distance $k > 1$ from the root in $G$, assign $d(v,r) \mod 3$ to
  $v$. This satisfies N\ref{rule:at-least-one-smaller}. Consider a
  triangle $uvw$, assume, without loss of generality, that $v$ is
  the closest vertex to $r$ among $u,v,w$, and let $k = d(v,r)$. Then
  $k \leq d(u,r) \leq k+1$ and $k \leq d(w,r) \leq k+1$, and thus
  N\ref{rule:bigger-smaller-not-touching} is satisfied. We have thus
  established the existence of an assignment that satisfies
  N\ref{rule:at-least-one-smaller}--N\ref{rule:bigger-smaller-not-touching}.
\end{normalproof}

Given a graph $G=(V,E)$ and a labeling $x_L: V \to \{0,1,2\}$, let
$\vec{G}$ be the directed graph defined by $V(\vec{G}) = V$ and
$E(\vec{G}) = \{\overrightarrow{vw} : vw \in E \text{ and } x_L(w) = x_L(v)-1
\mod 3 \}$.

\begin{lemma}\label{lem:acyclic}
  If
  N\ref{rule:at-least-one-smaller}--N\ref{rule:bigger-smaller-not-touching}
  are satisfied for every node in a dismantlable graph $G$, then
  $\vec{G}$ does not contain directed cycles.
\end{lemma}
\begin{normalproof}
  Let $G = (V,E)$ be a dismantlable graph and $x_L: V \to \{0,1,2\}$
  be a labeling of $G$ such that
  N\ref{rule:at-least-one-smaller}--N\ref{rule:bigger-smaller-not-touching}
  hold at every $v \in V$.  Consider a dismantling order
  $u_1, \ldots, u_n$ of $V$, and let the \emph{order} of a vertex $u$,
  be its position in the dismantling order. For any cycle
  $\sC = (v_0,v_1, \ldots, v_k)$ of $G$, the \emph{sum of orders} of
  $\sC$ is the sum of the orders of the vertices of $\sC$. By
  contradiction, assume that $\vec{G}$ contains a directed cycle, and
  among all directed cycles of $\vec{G}$, let
  $\sC = (v_0,v_1,\ldots,v_k)$ be a cycle that has a minimum sum of
  orders.
  Note that for any $0 \leq i \leq k$, we have
  $x_L(v_{i+1}) = x_L(v_i) -1 \mod 3$ (with the convention that
  $v_{k+1}= v_0$).  Without loss of generality, assume that $v_1$ is
  the node with the maximum order in $\sC$ and that $x_L(v_1) =
  1$. Let $u_1$ be a node dominating $v_1$ in the dismantling order of
  $G$. Thus, the order of $u_1$ is smaller than the order of
  $v_1$. Since $v_0 \in N(v_1) \subseteq N[u_1]$, either $u_1 = v_0$
  and $x_L(u_1) = x_L(v_0) = x_L(v_1) + 1 \mod 3 = 2$, or
  $x_L(u_1) \in \{1,2\}$ by N\ref{rule:bigger-smaller-not-touching}
  applied to the triangle $v_0v_1u_1$. Similarly, since
  $v_2 \in N(v_1) \subseteq N[u_1]$, we obtain that
  $x_L(u_1) \in \{0,1\}$. Consequently, $x_L(u_1) = 1$,
  $u_1 \notin \{v_0,v_2\}$, and $\overrightarrow{v_0u_1}$ and
  $\overrightarrow{u_1v_2}$ are arcs of $\vec{G}$. If
  $u_1 \notin \sC$, then $\sC' = (v_0,u_1,v_2, \ldots, v_k)$ is a
  directed cycle of $\vec{G}$ that has a smaller sum of orders than
  $\sC$, a contradiction. Suppose now that $u_1 = v_i$ with
  $3 \leq i \leq k$ and observe that the cycle
  $\sC'' = (v_0,u_1=v_i,v_{i+1}, \ldots, v_k)$ is directed cycle of
  $\vec{G}$ that has a smaller sum of orders than $\sC$. This shows
  that $\sC$ does not exist and ends the proof of the lemma.
\end{normalproof}

Let $T$ be a spanning subgraph of $\vec{G}$ in which each node $v$,
except the root, arbitrarily selects as a parent a neighbor
$w \in N(v)$ such that $x_L(w) = x_L(v) - 1 \mod 3$. We show that $T$
is a tree. Since $\vec{G}$ is acyclic by Lemma \ref{lem:acyclic}, $T$
is also acyclic. We  prove that $T$ is connected.

\begin{lemma}\label{lem:connected}
    If N\ref{rule:at-least-one-smaller}--N\ref{rule:bigger-smaller-not-touching} are satisfied for every node in a dismantlable graph, $T$ is connected. 
\end{lemma}
\begin{normalproof} 
  Let $u$ be a node in $T$. Since from Lemma \ref{lem:acyclic},
  $\vec{G}$ and as a result $T$ does not contain cycles, the directed
  path that starts at $u$ and follows the selected arcs must end at a
  node without a successor in $T$. From
  N\ref{rule:at-least-one-smaller}, the root is the only node without
  a successor. Hence, any node $u$ must be in a path that reaches the
  root, and as a result $T$ is connected.
\end{normalproof}

Since every node in $G$ has a label and must be in $T$, from Lemmas \ref{lem:acyclic} and \ref{lem:connected} we obtain:

\begin{theorem}[Soundness]\label{th:spanning-tree}
    In a dismantlable graph, if N\ref{rule:at-least-one-smaller}--N\ref{rule:bigger-smaller-not-touching} are satisfied for every node, $T$ is a spanning tree.  
\end{theorem}

\subsection{Completeness for Leader Certification in Dismantlable Graphs}\label{sec:completeness-dism}
In this section we show how our scheme can be extended to certify arbitrary sinks in dismantlable graphs. In \Cref{sec:elec-dism} we gave a labeling scheme that implies a unique sink defined to be a leader, with labels encoding the orientation of all edges incident to every node. 
  Observe that with our labeling schemes for leader election in
  chordal graphs and in trees, any node in the graph can be the sink
  and thus designated as a leader. Indeed, any chordal graph $G$
  contains at least two simplicial vertices, and thus for any vertex
  $u$, there exists a simplicial order $v_1,\ldots, v_n$ with $v_1=u$. By
  \Cref{obs:chordal-labeling-exists}, there exists then an orientation
  of $\vec{G}$ satisfying C\ref{rule:chordal-all-edges-directed}--C\ref{rule:chordal-no-cyclic-triangles}
  where $u$ is the unique sink of $\vec{G}$. In other words, the labeling schemes for trees and chordal graphs satisfy Completeness.
  This property is however no longer true for dismantlable
  graphs. Consider for example the graph $G$ of Figure
  \ref{fig:dismantlable} and observe that $a$ is the only dominated
  vertex of $G$. This implies that $a$ cannot be the sink of any
  orientation satisfying
  D\ref{rule:all-edges-directed}--D\ref{rule:no-cyclic-triangles}
  obtained by \Cref{obs:dismantlable-labeling-exists}. In fact, one
  can show that in any orientation $\vec{G}$ of $G$ satisfying
  D\ref{rule:all-edges-directed}--D\ref{rule:no-cyclic-triangles}, $a$
  is not a sink. Indeed, suppose there exists such an
  orientation. From D\ref{rule:dominating-neighbor} at $b$ and $c$,
  the edges $bd$ and $ce$ should be oriented as $\overrightarrow{db}$
  and $\overrightarrow{ec}$. Consider now the edge $de$ and assume
  without loss of generality that it is oriented as
  $\overrightarrow{de}$. Then $d$ has two non-adjacent outgoing
  neighbors $b$ and $e$ that do not have a common neighbor, and thus
  D\ref{rule:dominating-neighbor} does not hold at $d$. Consequently, the labeling scheme of \Cref{sec:elec-dism} satisfies Existence and Soundness, but not Completeness.

  However, by combining our labeling schemes, we can easily build a
  new labeling scheme that allows to elect any node in a dismantlable
  graph $G$. First, we use the labeling scheme from
  Section~\ref{sec:elec-dism} to mark a root $r$ in $G$. Then, we use the
  labeling scheme from Section~\ref{sec:st-dism} to build a spanning
  tree $T$ of $G$ rooted at $r$. Finally, considering only the edges
  of $T$ as unoriented edges (and forgetting the root $r$), we use the
  labeling scheme from Section~\ref{sec:trees} to elect a leader in
  $T$. As for any vertex $v$ and any spanning tree $T$ of $G$, there
  exists a valid orientation $\vec{T}$ of $T$ where $v$ is the unique
  sink of $\vec{T}$, we deduce that we can construct a locally
  checkable labeling solving leader election in $G$, where for each
  node $v$, there is a valid labeling where $v$ is the designated
  leader.

  \section{From Labeling to Self--Stabilization} \label{sec:gf}
A distributed self-stabilizing algorithm can autonomously recover from transient fault and is usually defined as a set of rules evaluated at every node. Each rule is of the form:

\centerline{$\texttt{name-of-rule}: \langle guard \rangle \longrightarrow command$}

\noindent where \emph{guard} is a boolean predicate over the variables of the closed neighborhood of the node and the \emph{command} specifies how the node updates its variables. Here, we assume that each node only updates its state variable. A self-stabilizing algorithm is said to be \emph{silent} if after some finite time no node updates its variables anymore. When activated, each node can perform computations based on its state and the states of its neighbors and update its own state. The order of node activations is determined by a \emph{scheduler}. A scheduler is characterized both by synchronicity and fairness. In particular, here we consider an \emph{asynchronous} scheduler meaning that at each step, an arbitrary (non-empty) subset of nodes may be activated. In terms of fairness, we assume a \emph{Gouda fair} scheduler~\cite{Gouda2001theory}.  According to \cite{dubois2011taxonomy}, Gouda fairness is the strongest fairness assumption. It requires that for any configuration that appears infinitely often during the execution of an algorithm, every configuration that is reachable by executing any non-empty subset of enabled commands (i.e., commands for which the guard is true) also appears infinitely often. The graphs we consider are finite. In this setting, Gouda fairness ensures that any configuration that is infinitely often reachable is eventually reached. In this section we show how to obtain a silent self--stabilizing algorithm that works under a Gouda fair scheduler, given a labeling scheme of finite type. 

We say that under a given algorithm $A$, a label assignment is \emph{final} if no node satisfies any of the guards in $A$. We define a correct final label assignment to be one that is both final with respect to an algorithm $A$ and correct with respect to a labeling scheme. Call $L$ the set of states in a given labeling. Informally, our algorithm works as follows. Every node that detects a constraint of the labeling scheme being violated sets its state to a \textit{reset} state that we denote $\bot$. Every node that has at least one neighbor in a reset state, also changes its state to the \textit{reset} state. Hence, an error detection is propagated throughout the graph from node to neighboring node. When a node that is at the \textit{reset} state is activated, it non-deterministically switches to a state from $L$. Gouda fairness ensures that the graph eventually reaches a correct global configuration.

 Given a labeling scheme $(L,R)$, a graph $G$ and a vertex
$v \in G$, a \emph{possible state} for $v$ in $L$ is a state $s$ from $L$
such that there exists a 1-ball $(B,x) \in R$ where the center of $B$
is labeled by $x(v) = s$ and $B$ is isomorphic to $B_1(v)$. If $(L,R)$
is a labeling scheme of finite type, then the set of possible states for each node in any graph $G$ is finite.

A node executing Algorithm~\ref{algo:gcssa} evaluates the guards in order. If the first guard is enabled, the corresponding command is performed and the second guard is not evaluated. 
\begin{algorithm}[htpb]  
    \SetKwInput{input}{Input}  
    \caption{General Labeling to Self-Stabilizing Algorithm}    
    \label{algo:gcssa}
    {\small \input{The labeling scheme $(L,R)$ and a label assignment $x_L : V \rightarrow L$}}
    \BlankLine    
    \nonl\underline{Every time node $v$ is activated:} \\ 
    \tcc{states from $L$ are always chosen non-deterministically}    
    \nonl\resizebox{.925\textwidth}{!}{%
    \begin{tabular}{llll}        
        \hypertarget{line-gcssa:c1}{S1} : & \emph{v.state $\neq \bot$ $\land$ $((B_1(v),x_C |_{N[v]}) \notin R$ $\lor$ $\exists u \in N(v):$ u.state = $\bot)$} & $\rightarrow$ & \emph{v.state} $\gets \bot$ \\

        \hypertarget{line-gcssa:c2}{S2} : & \emph{v.state = $\bot $} & $\rightarrow$ & \emph{v.state} $\gets$ \emph{possible state from} $L$    
    \end{tabular}%
    }
\end{algorithm}

\begin{theorem}\label{th:gcssa-correctness}    
  Let $(L,R)$ be a  labeling scheme of finite type
  satisfying the Existence and Soundness conditions for some task
  $\mathcal{T}$ in a connected graph $G$.  Starting from any 
  configuration, any Gouda fair execution of Alg.~\ref{algo:gcssa} on $G$ stabilizes to a correct final configuration.
\end{theorem}
\begin{normalproof}
  Let $G^*$ be a configuration where no node satisfies the guard of
  \hyperlink{line-gcssa:c1}{S1} or of \hyperlink{line-gcssa:c2}{S2},
  that is, $G^*$ is a final correct configuration. By the soundness
  property of $(L,R)$ for $\mathcal{T}$, $\mathcal{T}$ is solved in
  $G^*$. Consequently, in any final configuration, $\mathcal{T}$ is
  solved.

  We now show that any execution eventually reaches a final
  configuration. Since $(L,R)$ is of finite type, for each node $v$,
  the set of possible states is finite and thus in any execution,
  every node $v$ enters only in a finite number of different
  states. Since $G$ is finite, the number of different configurations
  existing in any execution (finite or infinite) is also finite.  By
  the existence property of $(L,R)$ for $\mathcal{T}$ in $G$, there
  exists a configuration $G^*=(G,x_L)$ of $G$ such that in $G$, every
  1-ball $(B_1(v),x_L |_{N[v]})$ belongs to $R$. Suppose there exists
  an infinite execution of the algorithm. This execution does not
  reach any final configuration, i.e., at each step, there is a node
  that can execute \hyperlink{line-gcssa:c1}{S1} or
  \hyperlink{line-gcssa:c2}{S2}. In the second case, it means that
  there is node that is labeled by $\bot$. In the first case, it means
  that one can reach in one step a configuration containing at least
  one node labeled $\bot$. Starting from a configuration containing
  one node labeled $\bot$, there exists a sequence of activation of
  nodes that enables reaching the configuration $G_\bot$ where all
  nodes are in state $\bot$. Then by Gouda fairness, in any
  infinite execution, the configuration $G_\bot$ is reached infinitely
  often. From $G_\bot$, there exists a sequence of node activations
  that enables reaching the configuration $G^*$. Again, by Gouda
  fairness, this implies that any infinite execution eventually
  reaches $G^*$. Since $G^*$ is a final configuration where no node
  can be activated, we obtain a contradiction with the existence of an
  infinite execution.
\end{normalproof}

A self--stabilizing algorithm can be obtained by combining the respective labeling scheme of \Cref{sec:trees,sec:chordal,sec:k4-free} with Algorithm \ref{algo:gcssa}.

\section{Open Problems}
  A natural direction for further research is to consider graph classes that generalize dismantlable graphs.  For example,
  \cite{fieux2020hierarchy} gives a hierarchy on classes (based also
  on the work of \cite{mazurkiewicz2007locally}) that have dismantling
  orders, called \emph{s-dismantlable}, where 1-dismantlable graphs
  correspond to the dismantlable
  graphs. 
  Informally, for $s > 1$, \cite{fieux2020hierarchy} defines a node to
  be $s+1$ dismantlable if the graph induced by its open neighborhood
  is $s$-dismantlable and a graph to be $s$-dismantlable if it is
  reducible to a single vertex by successive deletions of
  $s$-dismantlable vertices.

Finally, a potential direction for future research is exploring whether it is possible to automatically transform a labeling into a self-stabilizing algorithm with a weaker scheduler, as defined in the taxonomy of \cite{dubois2011taxonomy}, even using non-constant
memory. For example, the transformation from \cite{BlinFP14}, which uses logarithmic memory but works even with an unfair scheduler, contrasts with our algorithm, which uses constant memory but requires a Gouda fair scheduler. These two results suggest that there may be a trade-off between the strength of the scheduler and the memory requirements that has not yet been fully explored.

\bibliographystyle{splncs04}
\bibliography{biblio}

\end{document}